\def\withcomments{1}

\documentclass[11pt,english,autoref,thm-restate,numberwithinsect]{article}
\newcommand{\A}{{\mathcal A}}

\usepackage[ruled, noend, noline]{algorithm2e}
\usepackage{hyperref}
\usepackage{babel,comment}
\usepackage{amsmath, amsthm, amssymb, url}
\usepackage[numbers,longnamesfirst]{natbib}
\usepackage{verbatim}
\usepackage[usenames,dvipsnames]{color} 
\usepackage{multirow}
\usepackage{bigstrut}
\usepackage[disable]{todonotes}
\usepackage{cleveref,paralist}
\usepackage{fullpage, bbm}

\newtheorem{theorem}{Theorem}[section]
\newtheorem{lemma}[theorem]{Lemma}

\newtheorem{claim}[theorem]{Claim}

\newtheorem{definition}{Definition}[section]

\numberwithin{figure}{section}

\newcommand{\nA}{\mathbf{a}}
\newcommand{\nB}{\mathbf{b}}

\newcommand{\cP}{{\mathcal P}}

\newcommand{\C}{{\mathcal C}}
\newcommand{\eps}{{\epsilon}}

\newcommand{\dis}{dist}
\newcommand{\Dis}{Dist}

\newcommand{\ste}{{\epsilon\emph{-tester}}}

\newcommand{\dP}{\mathcal P}
\newcommand{\dN}{\mathcal N}

\newcommand{\Accept}{\textbf{Accept}\xspace}

\newcommand{\reject}{\textbf{reject}\xspace}
\newcommand{\lc}{lc}
\newcommand{\elc}{elc}

\newcommand{\integerset}[1]{[0..{#1})}

\newcommand{\lind}{i}

\DeclareMathOperator*{\E}{\mathbb E}

\newenvironment{myproof}{\begin{proof}
}{
\end{proof}}

\ifnum\withcomments=1
   \newcommand{\mnote}[1]{{\color{red}\footnote{{\color{violet} {\bf M:} #1}}}}
   \newcommand{\pnote}[1]{{\color{green}\footnote{{\color{green} {\bf P:} #1}}}}
\newcommand{\srnote}[1]{{\color{blue}\footnote{{\color{blue} {\bf S:} #1}}}}
  
  \newcommand{\dr}[1]{{\color{orange} #1}}
  
  \newcommand{\mch}[1]{{\color{violet}#1}}
\else
   \newcommand{\mnote}[1]{}
   \newcommand{\pnote}[1]{}
  \newcommand{\srnote}[1]{}
   \newcommand{\changenote}[1]{}
   \newcommand{\dr}[1]{}
   \newcommand{\mch}[1]{}
\fi

\title{Testing Connectedness of Images\footnote{A preliminary version of this article~\cite{BermanMRR23} was published in the proceedings of the $27$th International Conference on Randomization and Computation, RANDOM, 2023.}}

\author{Piotr Berman\\
\and Meiram Murzabulatov \thanks{Nazarbayev University, Kazakhstan; {\tt meiram.murzabulatov@nu.edu.kz}. Funded and supported by the Science Committee of the Ministry of Science and Higher Education of the Republic of Kazakhstan (Grant IRN AP19679952). 
}\\
\and Sofya Raskhodnikova\thanks{Boston University, USA; {\tt sofya@bu.edu}. 
}\\
\and Dragos-Florian Ristache\thanks{Boston University, USA; {\tt dragosr@bu.edu}. Supported by Shibulal Family Career Development Grant.}
}

\begin{document}
\maketitle
\begin{abstract}
We investigate algorithms for testing whether an image is connected. Given a proximity parameter $\eps\in(0,1)$ and query access to a black-and-white image represented by an $n\times n$ matrix of Boolean pixel values, a (1-sided error) connectedness tester accepts if the image is connected and rejects with probability at least 2/3 if the image is $\eps$-far from connected. We show that connectedness can be tested nonadaptively with $O(\frac 1{\eps^2})$ queries and adaptively with  $O(\frac{1}{\eps^{3/2}} \sqrt{\log\frac{1}{\eps}})$ queries. The best connectedness tester to date, by Berman, Raskhodnikova, and Yaroslavtsev (STOC 2014) had query complexity  $O(\frac 1{\eps^2}\log \frac 1{\eps})$ and was adaptive. We also prove that every nonadaptive, 1-sided error tester for connectedness must make $\Omega(\frac 1\eps\log \frac 1\eps)$ queries. 

\end{abstract}
\section{Introduction}\label{sec:intro}

Connectedness is one of the most fundamental properties of images~\cite{MinskyP70}. In the context of property testing, it was first studied two decades ago \cite{Ras03}, but the query complexity of this property is still unresolved. We improve the algorithms for testing this property and also give the first lower bound on the query complexity of this task.

We focus on black-and-white images. For simplicity, we only consider square images, but everything in this paper can be easily generalized to rectangular images. We represent an image by an $n\times n$ binary matrix $M$ of pixel values, where 0 denotes white and 1 denotes black. To define {\em connectedness}, we consider {\em the  image graph $G_M$} of an image $M$. The vertices of $G_M$ are
$\{(i,j)\mid M[i,j]=1\}$, and two vertices $(i,j)$ and $(i',j')$ are connected by an edge if $|i-i'|+|j-j'|=1$.
In other words, the image
graph consists of black pixels connected by the grid lines. The image is
{\em connected} if its image graph is connected. The set of connected images is denoted $\C.$ 

We study connectedness in the property testing model~\cite{RS96,GGR98}, first considered in the context of images in~\cite{Ras03}. A (1-sided error) property tester for connectedness gets query access to the input matrix $M$.  Given a proximity parameter $\eps\in(0,1)$, the tester has to accept 
if $M$ is connected and reject with probability at least 2/3 if $M$ is $\eps$-far from connected. An image is $\eps$-far from connected if at least an $\eps$ fraction of pixels have to be changed to make it connected. The tester is {\em nonadaptive} if it makes all its queries before receiving any answers; otherwise, it is {\em adaptive.}

In \cite{Ras03}, it was shown that connectedness can be tested adaptively with $O(\frac 1{\eps^2}\log^2 \frac 1{\eps})$ queries. The complexity of testing connectedness adaptively was later improved to $O(\frac 1{\eps^2}\log \frac 1{\eps})$ in \cite{BermanRY14}.

\subsection{Our Results}\label{sec:results}
\paragraph{Connectedness Testers.} 

We give two new algorithms for testing connectedness of images: one adaptive, one nonadaptive. Both improve on the best connectedness tester to date in terms of query complexity. Previously, no nonadaptive testers for connectedness were proposed.

 \begin{theorem}\label{thm:connectedness_tester}
Given a proximity parameter $\eps\in(0,1)$, connectedness of $n\times n$ images, where $n\geq 8\eps^{-3/2}$, 
can be $\eps$-tested adaptively and with 1-sided error with query and time complexity $O(\frac{1}{\eps^{3/2}} \sqrt{\log\frac{1}{\eps}})$.

It can be tested nonadaptively and with 1-sided error  with query and time complexity $O(\frac{1}{\eps^{2}})$.
\end{theorem}

Previous algorithms for testing connectedness of images are modeled on the connectedness tester for bounded-degree graphs by Goldreich and Ron~\cite{GR02}: they pick a uniformly random pixel and adaptively try to find a small connected component by querying its neighbors. As discussed in \cite{Ras03}, even though connectedness of an image is defined in terms of the connectedness of the corresponding (degree-4) image graph, these two properties are different because of how the distance is defined. In the bounded-degree graph model, the (absolute) distance between graphs is the number
of edges that need to be changed to transform one graph into the other. In contrast, the (absolute) distance between two image graphs is the number of pixels (vertices) on which they differ; in other words, the edge
structure of the image graph is fixed, and only vertices can be added or removed to transform one graph into another. However, previous connectedness testers in the image model did not take advantage of the differences.

\begin{figure}[ht]
\centering
\begin{minipage}[b]{0.44\linewidth}
\includegraphics[width=0.9\linewidth]{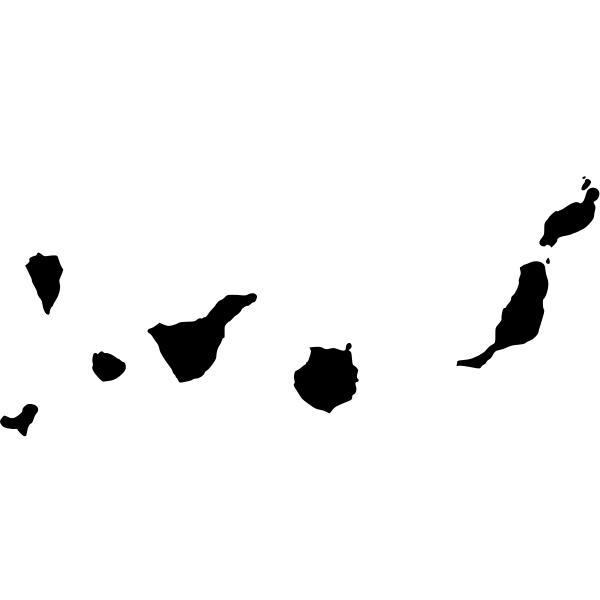}
\caption{An image $M$.}
\label{fig:islands}
\end{minipage}
\hspace{0.1\linewidth}
\begin{minipage}[b]{0.44\linewidth}
\includegraphics[width=0.9\linewidth]{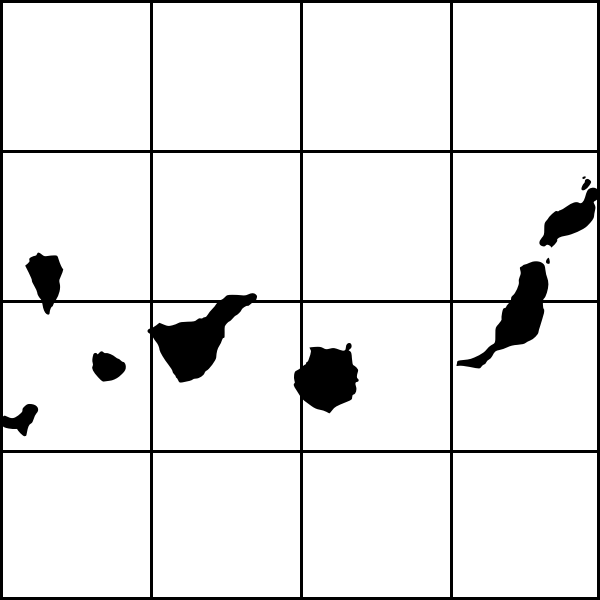}
\caption{The same image with a grid.}
\label{fig:islands-grided}
\end{minipage}
\end{figure}
As our starting point, we use an idea from \cite{BermanMR22}. That work gave an algorithm for approximating the (relative) distance to the nearest connected image within additive error $\eps$ using $O(\frac{1}{\eps^{4}})$ queries in $\exp\left(O\left(\frac 1 \eps \right)\right)$ time.
The authors
observed that one can modify an image in a small number of pixels
by drawing a grid on the image
(as shown in Figures~\ref{fig:islands} and~\ref{fig:islands-grided}). In the resulting image, the distance to connectedness is determined by the properties of individual squares into which the grid lines partition the image.

Our algorithms consider different partitions of the grid (a logarithmic number of partitions in $\frac 1 \eps$). For each partition, they sample random squares and try to check whether the squares satisfy the following property called {\em border-connectedness}.

\begin{definition}[Border-connectedness]\label{def:border_connectedness}
A (sub)image $s$ is {\em border-connected} if for every black pixel $(i,j)$ of $s$, the image graph $G_s$ contains a path from $(i,j)$ to a pixel on the border of $s$. The property
\emph{border-connectedness,} denoted $\C'$, is the set of all border-connected images.
\end{definition}
Our nonadaptive algorithm reads all pixels in each sampled square. Our adaptive algorithm further partitions each square into diamonds, as shown in \Cref{fig:diamond}.  (We could have partitioned into squares again, but partitioning into diamonds makes the proof cleaner and saves a constant factor in the analysis.)
The algorithm queries all pixels on the diamond lattice and then adaptively tries to catch a witness (that is, a small connected component) in one of the diamonds, using the lattice structure.
It partitions the diamonds into those that can be easily connected to the border and those that cannot. It tries to catch either a black pixel in the diamond in the latter group or, alternatively, a small connected component in the diamond in the former group by running a breadth-first search (BFS) for a random number of steps, drawn from a carefully selected distribution.

Even though we use the notion of border-connectedness from \cite{BermanMR22}, our algorithms and analysis are different. The authors of \cite{BermanMR22} approximate the distance to connectectedness by partitioning the image into squares with the side length $\Theta(1/\eps)$, sampling such a square, and then returning the exact distance to border-connectedness of the sampled square. (Computing this distance exactly, even in exponential time, is highly nontrivial.) In contrast, we use multiple partitions of the image. Moreover, our adaptive subroutine queries a small fraction of the pixels in the sampled square to certify that the square is not border-connected. The notions we use in the analysis of our algorithms and the accompanying correction procedures are also new.

\paragraph{The Lower Bound.} 
We also prove the first nontrivial lower bound for testing connectedness of images. Note that all nontrivial properties have query complexity  $\Omega(1/\eps)$, even for adaptive testers. 
In particular, for connectedness, this is the number of queries needed to distinguish between the white image and an image, where we color black a random subset of size $\eps n^2$ of the set of pixels with both coordinates divisible by 3.
By standard arguments, it implies a lower bound of $\Omega(1/\eps)$.
Some properties of images, such as being a halfplane, can be tested nonadaptively (with 1-sided error) with $O(1/\eps)$ queries \cite{BMR19-algorithmica}. We show that it is impossible for connectedness.

\begin{theorem}\label{thm:nonadap-bound-connectedness}
Every nonadaptive, 1-sided error $\eps$-tester for connectedness of images must query  $\Omega(\frac 1{\eps}\log \frac1 {\eps})$ pixels (for some family of images).
\end{theorem}
Every 1-sided error tester must catch a witness of disconnectedness in order to reject. This witness could include a connected component completely surrounded by white pixels. The difficulty for proving hardness is that, unlike in the case of finding a witness for disconnectedness of graphs, the algorithm does not have to read the entire connected component. Instead, it is sufficient to find a closed white loop with a black pixel inside it (and another black pixel outside it). As we discussed, it is enough
for an algorithm to look for witnesses inside relatively small squares (specifically, squares with side length $O(1/\eps)$), since adding a grid around such squares, as shown in \Cref{fig:islands-grided}, would change $O(\eps n^2)$ pixels. But no matter how a witness inside such a square looks like, it can be easily captured with $O(1/\eps)$ queries if the border of the square is white.

To overcome this difficulty, we consider a checkerboard-like pattern with white squares replaced by many parallel lines, called {\em bridges}, with one white (disconnecting) pixel positioned randomly on each bridge. See \Cref{fig:interesting-square}. To catch a white border around a connected component, a tester has to query all disconnecting pixels of at least one square. To make this difficult, we hide the checkerboard pattern inside a randomly positioned {\em interesting window}. The sizes of interesting windows and their positions are selected so that the tester cannot effectively reuse queries needed to succeed in catching the disconnecting pixels in each interesting window.

\subsection{Other Related Work}
In addition to \cite{GR02}, connectedness testing and approximating the number of connected components in graphs in sublinear time was explored in \cite{CRT05,BermanRY14,BerenbrinkKM14}. Other property testing tasks studied in the pixel model of images, the model considered in this paper, include testing whether an image is a half-plane~\cite{Ras03,Ras-thesis03,Mur-thesis17}, convexity~\cite{Ras03,Ras-thesis03,Mur-thesis17,BMR19uniform,BMR19-algorithmica}, and image partitioning properties~\cite{KleinerKNB11}. Early implementations and applications to vision were provided in \cite{KleinerKNB11,KormanRT,KormanRTA13,ParalCR19}.
Finally, general classes of matrix properties were investigated, including
matrix-poset properties~\cite{FischerN07}, {\em earthmover resilient} properties~\cite{BF18}, {\em hereditary} properties~\cite{AlonBF17}, and classes of matrices that are free of specified patterns~\cite{BKR17}.

Testing connectedness has also been studied by Ron and Tsur \cite{TsurR10} with a different input representation suitable for testing sparse images.

\section{Definitions and Notation}\label{sec:definitions_notation}
 We use $\integerset{n}$ to denote the set of integers $\{0,1,\ldots,n-1\}$ and $[n]$ to denote $\{1,2,\ldots,n\}$.
By $\log$ we mean the logarithm base 2.
For a set $S\subset \integerset{n}^2$ and $(i,j)\in \integerset{n}^2$, we define
$S+(i,j)=\{(x+i,y+j):~(x,y)\in S\}$.

\paragraph{Image Representation.}
We represent an image by an $n\times n$ binary matrix $M$ of pixel values, where 0 denotes white and 1 denotes black.
The object is a subset of $\integerset{n}^2$
corresponding to black pixels; namely, $\{(i,j)\mid M[i,j]=1\}$.
The {\em border of the image} is the set $\{(i,j)\in \integerset{n}^2 \mid i\in\{0,n-1\} \text{ or } j \in\{0,n-1\}\}$.

\paragraph{Property Testing Definitions.} 
A {\em property} $\cP$ is a set of images. We say that an image $M$ {\em satisfies} $\cP$ if $M\in\cP$; otherwise, $M$ {\em violates} $\cP.$ The {\em absolute distance} from an image $M$ to a property $\cP$, denoted $\Dis(M,\cP),$ is the smallest number of pixels in $M$ that need to be modified to get an image in $\cP.$
The (relative) distance between an $n\times n$ image $M$ and a property $\cP$ is $\dis(M,\cP)=\Dis(M,\cP)/n^2.$ We say that $M$ is $\eps$-far from $\cP$ if $\dis(M,\cP)\geq \eps$; otherwise, $M$ is $\eps$-close to $\cP.$
    
\section{Adaptive and Nonadaptive Property Testers for Connectedness}\label{tester_for_connectedness}
In this section, we present our testers for connectedness, proving Theorem~\ref{thm:connectedness_tester}. Both testers use the same top-level procedure, described in Algorithm~\ref{alg:connectedness_tester}. First, it samples random pixels to ensure that a black pixel is found. 
The black pixel will be used later together with an isolated black component to certify non-connectedness.
Then Algorithm~\ref{alg:connectedness_tester} considers a logarithmic number of partitions of the image into subimages of the same size. For each partition, it samples a carefully selected number of these
subimages and tests them for border-connectedness (see Definition~
\ref{def:border_connectedness}). This is where the two algorithms diverge. The nonadaptive algorithm tests for border-connectedness using the subroutine \emph{Exhaustive-Square-Tester} which queries all pixels in the sampled square and determines exactly if the square is border-connected. This can be done by a BFS on the image graph for the sampled square that checks if all the black pixels can be reached from the border.
 The adaptive algorithm uses subroutine \emph{Diagonal-Square-Tester} (Algorithm~\ref{alg:diagonal_square_tester}).
If the top-level procedure finds a subimage that violates border-connectedness and a black pixel outside that subimage, it rejects; otherwise, it accepts.

To simplify the analysis of the algorithm, we assume\footnote{\label{footnote:integrality-issues}This assumption can be made
w.l.o.g.\ because if $n\in(2^{i-1}+1,2^{i}+1)$ for some $\lind$ , instead of the original image $M$ we can consider a $(2^{i}+1)\times (2^{i}+1)$ image $M'$, which is equal to $M$ on the corresponding coordinates and has white pixels everywhere else. Let
$\eps'=\eps n^{2}/(2^{i}+1)^{2}$.
To $\eps$-test $M$ for connectedness, it suffices to $\eps'$-test $M'$ for connectedness.
The resulting tester for $M$ has the desired query complexity because $\eps'=\Theta(\eps)$. If $\eps\in(1/2^{j},1/2^{j-1})$ for
some $j$, to $\eps$-test a property ${\cal P}$, it suffices to run an $\eps''$-test for ${\cal P}$ with  $\eps^{\prime\prime}=1/2^{j}<\eps$.} that $n-1$ and $1/\eps$ are powers of
$2$.
Next, we define terminology used to describe the partitions considered by Algorithm~\ref{alg:connectedness_tester}.

\begin{figure}[ht]
\centering
\includegraphics[width=0.3\linewidth]{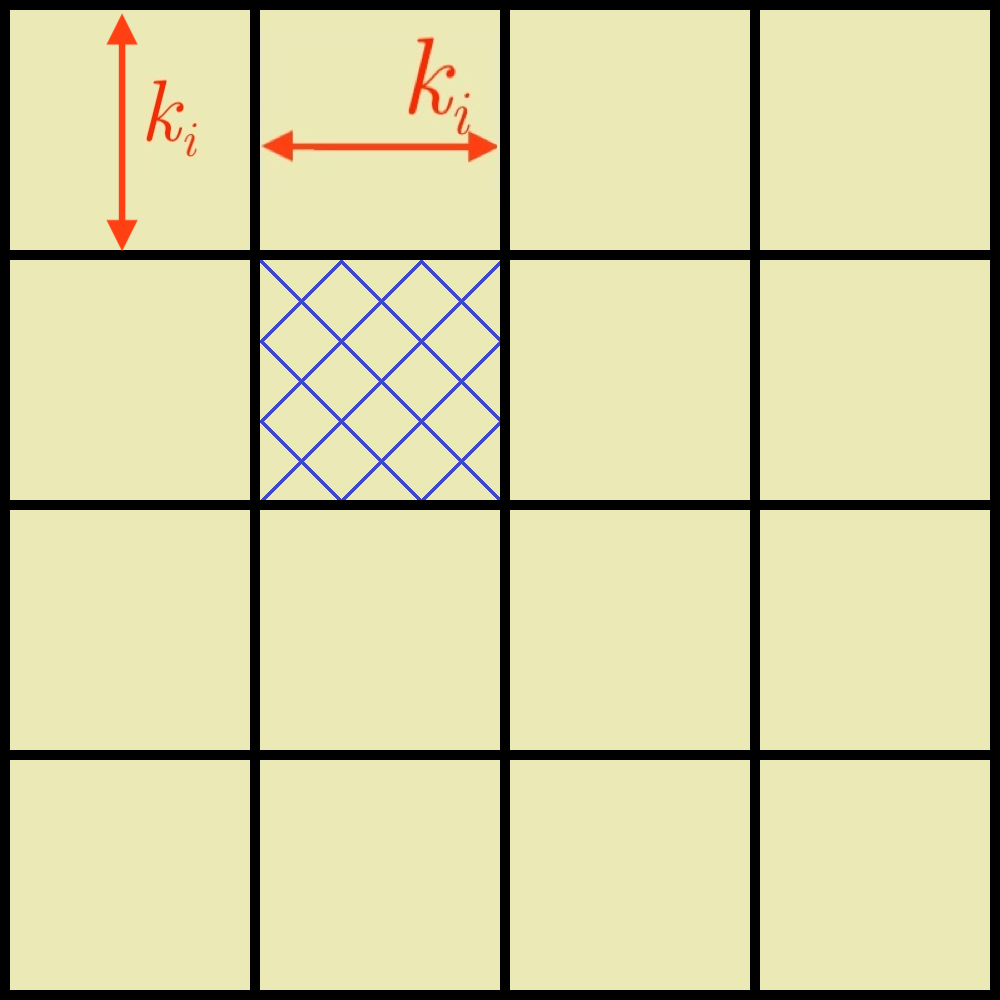}
\caption{An illustration to \Cref{def:Grid_pixels_squares_of_different_levels}: black lines consist of grid pixels; the 16 yellow $k_i\times k_i$ squares represent squares of $S_i$. One of the squares includes diagonal lattice pixels from \Cref{def:diagonal_lattice_pixels_and_regions} that are used in \Cref{alg:diagonal_square_tester}.}
\label{fig:gridImage}
\end{figure}

\begin{definition}[Grid pixels, squares of different levels, witnesses]
\label{def:Grid_pixels_squares_of_different_levels}
For $i\in\integerset{\log\frac{1}{\eps}}$,
let $k_{i}=\frac{4}{\eps}\cdot
2^{-i}-1$.
Call $(x,y)$ {\em a grid pixel of level $i$}  if  $(k_{i}+1)|x$ or $(k_{i}+1)|y$.
For all coordinates $u,v$, which are divisible by 
$k_i+1$,
the $k_{i}\times k_{i}$ subimage that consists of pixels $[k_{i}]^{2}+(u,v)$ is
called a \emph{square of level} $\lind$. The set of all squares of level $\lind$  is denoted ${\cal{S}}_{i}$.
{\em Boundary pixels} of a square of level $\lind$  are the pixels of the square which are adjacent to the grid pixels of level $\lind$. A square of any level 
that violates property $\C'$ (see Definition~\ref{def:border_connectedness}) is called a \emph{witness}.
\end{definition}
\SetAlCapHSkip{0pt} 
\setlength{\algomargin}{0em}
\begin{algorithm}
\caption{$\ste$ for connectedness.}
\label{alg:connectedness_tester}

\SetKwInOut{Input}{input}\SetKwInOut{Output}{output}
\Input{parameter $\eps\in(0,1)$; access to a $n\times n$ binary
matrix $M$.}
\DontPrintSemicolon
\BlankLine
\nl\label{step:query-pixels}Query $\frac{8}{\eps}$ pixels uniformly at random with replacement.
\;
\nl\label{step:go-through-all-i}\For {$i=0$ {\bf to}  $\log\frac{1}{\eps}-1$}{
\nl \label{step:square-test}{\bf Repeat}  $2^{i+1}$ times:
 \begin{enumerate}[(a)]
     \item\label{step:sample-squares} 
     Sample a uniformly  random  square $s$ of level $\lind$  (see Definition~\ref{def:Grid_pixels_squares_of_different_levels}) and let $[k_{i}]^{2}+(u,v)$ be the set of its pixels.

 \item\label{step:test-for-border-connectedness} 
 
Run the border-connectedness subroutine  
on the square $s$:
if the tester is nonadaptive, use \emph{Exhaustive-Square-Tester} (see the beginning of Section~\ref{tester_for_connectedness}); otherwise, use \emph{Diagonal-Square-Tester} (Algorithm~\ref{alg:diagonal_square_tester}) with inputs $i,u,v$. \\{\bf If} the subroutine rejects and Step~\ref{step:query-pixels} detected a black pixel outside $s$, \reject.
\;
\end{enumerate}
}
\nl
\Accept.
\;
\end{algorithm}

\subsection{Effective Local Cost and the Structural Lemma}
In this section, we state and prove the main structural lemma 
(Lemma~\ref{lem:sum_of_local_costs})
used in the analysis of Algorithm~\ref{alg:connectedness_tester}. It relates the distance to connectedness to the properties of individual squares, defined next.

\begin{definition}[Local cost and effective local cost]\label{def:elc}
For a level $i$, consider a square $s\in {\cal{S}}_{i}$.
The \emph{local cost} of $s$ is $\lc(s)=\Dis(s,\C')$. The \emph{effective local cost} of $s$ is
$\elc(s)=\min(2k_{i},\lc(s))$.
\end{definition}

Next we state and prove two claims used in the proof of Lemma~\ref{lem:sum_of_local_costs}.
\begin{claim}
\label{cl:parent_children_cost_relation}
For any square $s$ of level $i\in \integerset{\log\frac{1}{\eps}-1}$, let $ch(s)$ denote the set of its $4$ children (i.e.,
squares of level $i+1$ inside it).
Then $\lc(s)\leq \elc(s)+\sum\nolimits_{q\in ch(s)} \lc(q)$.
\end{claim}

\begin{proof}
If $\lc(s)\leq 2k_{i}$ then $\elc(s)=\lc(s)$. Since all costs are nonnegative,
the inequality in \Cref{cl:parent_children_cost_relation} becomes trivial.

Now assume that $\lc(s)>2k_{i}$. Then $\elc(s)=2k_{i}$. We can modify
$\sum\nolimits_{q\in ch(s)} \lc(q)$ pixels in $s$ so that all its children
satisfy the property $\C'$. (Note that here $\C'$ is  the set of $k_i\times k_i$ (sub)images.) Then we can make black all pixels of $s$ that
partition it into its children, i.e., pixels $\{(x,y)\mid
x=\frac{k_{i}+1}{2}$ or $y=\frac{k_{i}+1}{2}\}$. There are at most $2k_{i}$ such
pixels, and after this modification $s$ will satisfy $\C'$. Hence,
$\lc(s)\leq \elc(s)+\sum\nolimits_{q\in ch(s)} \lc(q)$.
\end{proof}

\begin{claim}[Distance to border-connectedness]
\label{cl:max_dist_to_border_connectedness}
Let $s$ be a $k\times k$ image. Then $$\Dis(s,\C')\leq
\frac{k^{2}}{4}.$$
\end{claim}
\begin{proof}
 If $s$ contains at most $\frac{k^{2}}{4}$ black pixels, we can make all of them white, i.e.,
 modify at most $\frac{k^{2}}{4}$ pixels and obtain an image that satisfies $\C'$.
 Now consider an image $s$ with more than  $\frac{k^{2}}{4}$ black pixels, i.e., with less than $\frac{3k^{2}}{4}$ white pixels. Partition all
 pixels of $s$ into $3$ groups such that group $i\in\{0,1,2\}$ contains all pixels $(x,y)$, where $y\equiv i\pmod 3$. Making all pixels of one group black produces an image where every third row is completely black. Note that completely black rows are also connected to the border. Every black pixel in the recolored image that is not in such a row is either directly above or directly below a black row and, consequently, also connected to the border. Thus, the new image satisfies $\C'$. By averaging, at least one group has less than  $\frac{k^{2}}{4}$ white pixels. Making all these white pixels black results in an image that satisfies $\C'$, proving the claim.
\end{proof}

\begin{lemma}[Structural lemma]
\label{lem:sum_of_local_costs}
Let $M$ be an $n\times n$ image that is $\eps$-far from $\C$. Then the sum of
effective local costs of all squares of all levels inside $M$ is at least $\frac{\eps n^{2}}{2}$.
\end{lemma}
\begin{proof}
To obtain a connected image, we can make all the $\frac{\eps n^{2}}{2}$ grid pixels of level $0$ black and modify pixels inside every square of $\mathcal{S}_{0}$ to  ensure it satisfies the property $\C'$.
Thus, $$\sum\nolimits_{s\in \mathcal{S}_{0}} \lc(s)\geq \Dis(M,\C)-\frac{\eps
n^{2}}{2}\geq\frac{\eps n^{2}}{2}.$$
Consequently,
it suffices to show that
$\sum\nolimits_{i=0}^{\log\frac{1}{\eps}-1}\sum\nolimits_{s\in \mathcal{S}_{i}}
\elc(s)\geq\sum\nolimits_{s\in \mathcal{S}_{0}} \lc(s)$.

Fix a level $i$ and a square $s\in \mathcal{S}_{i}$.
For a level $j\geq i$, let $desc(s,j)$  denote the set of all squares of level $j$ inside $s$.
(In particular, $desc(s,i)$ contains
only $s$.)
 We will prove by induction that for every integer $j\in
[i,\log\frac{1}{\eps}-1),$
\begin{align}\label{eq:lc-bound}
\lc(s)\leq\sum_{h=i}^{j}\sum\nolimits_{q\in
desc(s,h)}\elc(q)+\sum\nolimits_{q\in desc(s,j+1)}\lc(q).    
\end{align}

For $j=i$ (base case), the inequality in \eqref{eq:lc-bound} holds since it is equivalent to the
statement in Claim~\ref{cl:parent_children_cost_relation}. Assume
that \eqref{eq:lc-bound} holds for $j=m$, that is,
$$\lc(s)\leq\sum_{h=i}^{m}\sum\nolimits_{q\in
desc(s,h)}\elc(q)+\sum\nolimits_{q\in desc(s,m+1)}\lc(q).
$$
 We will prove \eqref{eq:lc-bound} holds  for $j=m+1$. By Claim~\ref{cl:parent_children_cost_relation},
$$\lc(q)\leq
\elc(q)+\sum\nolimits_{f\in ch(q)}\lc(f).$$ Thus,
$$\sum\nolimits_{q\in
desc(s,m+1)}\lc(q)\leq \sum\nolimits_{q\in
desc(s,m+1)}\elc(q)+\sum\nolimits_{q\in desc(s,m+2)}\lc(q).$$  By the induction hypothesis,
\begin{align*}
    \lc(s) &\leq\sum_{h=i}^{m}\sum\nolimits_{q\in
desc(s,h)}\elc(q)+\sum\nolimits_{q\in desc(s,m+1)}\lc(q)\\
&\leq
\sum_{h=i}^{m+1}\sum\nolimits_{q\in desc(s,h)}\elc(q)+\sum\nolimits_{q\in
desc(s,m+2)}\lc(q),
\end{align*}
completing the inductive argument.

By \eqref{eq:lc-bound} applied with $j=\log \frac{1}{\eps}-2$, we get that for every square $s$ of level $i$,
\begin{align}
\lc(s)
&\leq\sum_{h=i}^{\log \frac{1}{\eps}-2}\sum\nolimits_{q\in
desc(s,h)}\elc(q)+\sum\nolimits_{q\in desc(s,\log\frac{1}{\eps}-1)}\lc(q)\nonumber\\
&=\sum_{h=i}^{\log\frac{1}{\eps}-1}\sum\nolimits_{q\in desc(s,h)}\elc(q),
\label{eq:lc-one-square}
\end{align}
where the final equality holds because in each square of level $i=\log \frac{1}{\eps}-1$, we have $k_i=\frac{4}{\eps}\cdot
2^{-i}-1=7$, and consequently,  by 
\Cref{cl:max_dist_to_border_connectedness},
the local cost is at most $\frac{7^2}{4}<2\cdot 7$, i.e., it is equal to the effective local cost of that square.
Summing up \eqref{eq:lc-one-square} for all squares in $\mathcal{S}_{0}$, we get
$$\sum\nolimits_{s\in \mathcal{S}_{0}}\lc(s)\leq
\sum\nolimits_{s\in
\mathcal{S}_{0}}\sum\nolimits_{h=0}^{\log\frac{1}{\eps}-1}\sum\nolimits_{q\in
desc(s,h)}\elc(q)=\sum\nolimits_{h=0}^{\log\frac{1}{\eps}-1}\sum\nolimits_{s\in
\mathcal{S}_h}\elc(s),$$
where the last equality is obtained by switching the order of summations and rearranging the second summation in terms of the levels.
\end{proof}

\subsection{Testing Border-Connectedness}
In this section, we state and analyze our adaptive border-connectedness subroutine after defining the concepts used in it. We state the guarantees of both border-connectedness subroutines in \Cref{lem:border-connectedness-success_probability}.

For the adaptive subroutine, we partition the square into diamonds and fences surrounding them, as described in \Cref{def:diagonal_lattice_pixels_and_regions}. The subroutine queries all pixels on the fences and categorizes diamonds into those whose black pixels are potentially connected to the border (set $B$ in \Cref{alg:diagonal_square_tester}) and those whose black pixels are definitely not (set $A$  in \Cref{alg:diagonal_square_tester}).  Then it tries to find a black pixel in a diamond from set $A$ or (using BFS) an isolated black component in diamonds from set $B$. Observe that either of the two provides evidence that the square is not border-connected.
In order to ensure that \Cref{alg:diagonal_square_tester} does not get stuck investigating a large connected component, BFS is stopped after a random number of steps, $x$. This idea was used for estimating the number of connected components in a graph by Berenbrink, Krayenhoff, and Mallmann-Trenn~\cite{BerenbrinkKM14}. However, they select $x$ from a different distribution: specifically, with $\Pr[x\geq j]= 1/j^2$, whereas we set $\Pr[x\geq j]= 1/j$. 

\begin{figure}
\centering
\includegraphics[width=0.35\linewidth]{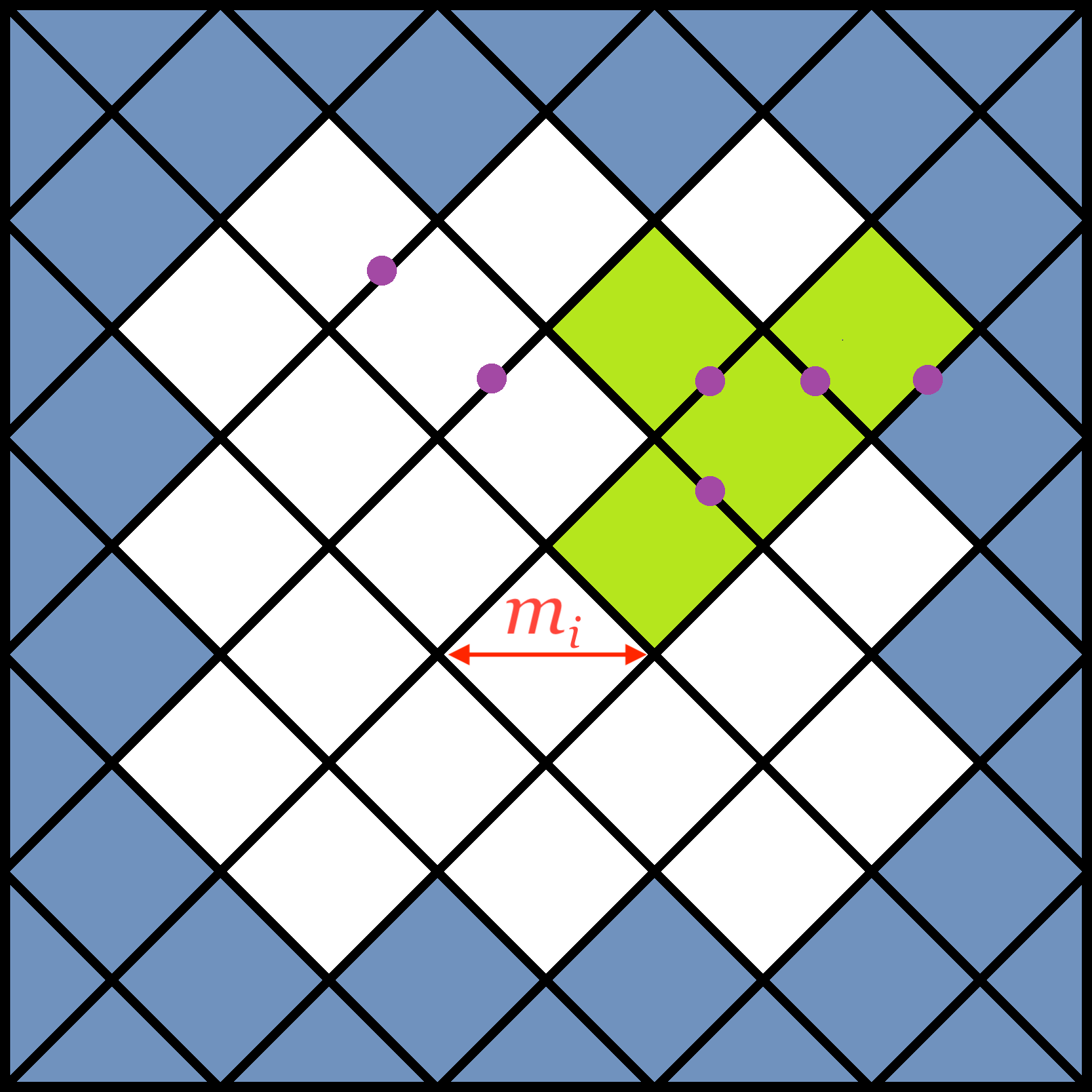}
\caption{
An example of execution of Algorithm 2. Black lines represent lattice pixels. The blue diamonds are included in $B$ because they contain pixels from the border of the square. The purple dots represent black pixels in the lattice. The green diamonds are added to $B$ during the BFS because their fences contain black pixels. The white diamonds remain in $A$.}
\label{fig:diamond}
\end{figure}

\begin{definition}[Diagonal lattice pixels, diamonds and fences]
\label{def:diagonal_lattice_pixels_and_regions}
For a fixed value of $i$, consider a square $s$ in $\mathcal{S}_{i}$. Let $m_{i}$ be the largest odd integer less than or equal to $\lceil\sqrt{k_{i}/\log k_{i}}\rceil$. \emph{Diagonal lattice pixels} of
the square is the set of pixels $L=\{(x,y)\in s\mid m_{i}|(x+y)$ or $m_{i}|(x-y)\}$. Let
$D$ be a $k_{i}\times k_{i}$ image whose pixels with coordinates from $[k_{i}]^{2}-L$ are black and the remaining pixels are white.
A set of pixels of the square whose corresponding pixels
in $D$ form a connected component is called a \emph{diamond} of the square.
A set of all diagonal lattice pixels that have some neighbouring
pixel(s) from a particular diamond is called the \emph{fence} of that diamond. 
\end{definition}
Observe that lattice pixels are not part of any diamond. Moreover, some diamonds (those that have pixels from the border of the square) are partial.

\begin{algorithm}
\caption{Border-connectedness subroutine \emph{Diagonal-Square-Tester}.}
\label{alg:diagonal_square_tester}
\SetKwInOut{Input}{input}\SetKwInOut{Output}{output}
\Input{level $i$, coordinates $u,v\in[n]$; access to
 an $n\times n$ matrix $M$.}
\DontPrintSemicolon
\BlankLine
Let $s$ be a square of level $i$ that consists of pixels $[k_{i}]^{2}+(u,v)$ and
$m_{i}$ be the largest odd integer less than or equal to $\lceil\sqrt{k_{i}/\log k_{i}}\rceil$.\;
\nl\label{step:diagonal lattice pixels}Query all the diagonal lattice pixels of $s$
(see
Definition~\ref{def:diagonal_lattice_pixels_and_regions}). \; 
\nl
Initialize $B$ to be the set of all diamonds of $s$ that contain
a border pixel of $s$. Initialize $A$ to be the set of the remaining diamonds of $s.$ 
\;
\nl
\label{step:alg-2-step-3}While $\exists d_1\in B$ and $\exists d_2\in A$ such that $d_1$ and
$d_2$ have a black pixel in the common portion of their fences, 
move $d_{2}$ from $A$ to $B$.\;
  \nl\label{step:sample-before-bfs}{\bf Repeat} $\frac1 2 k_{i}m_{i}$ 
times:
\begin{enumerate}[(a)]
    \item\label{step-sample} Sample a uniform pixel $p$ from the square $s$.
    \item\label{step:black-discovered} {\bf If} $p$ is a black pixel from a diamond in $A$ or its fence, \reject.
\item\label{step:BFS}  {\bf If} 
$p$ is
a black pixel from a diamond in $B$, 
pick a natural number $x\in[k_{i}^{2}]$ from the distribution with the probability
  mass function $f(j)=\frac {1}{j(j+1)}$ for all $j\in [k^2_i-1]$ and $f(j)=\frac {1}{j}$ for $j=k^2_i$. (Observe that $\Pr[x\geq j]=\frac{1}{j}$ for all $j\in[k_{i}^{2}]$.)
Starting from 
$p$, perform a BFS  of its connected component, halting if $x+1$ black pixels
in $p$'s component are discovered {\bf or} a black pixel on the border of square $s$ is reached.
{\bf If} the BFS  halts after discovering at most $x$ black pixels, 
  none of which are on the border of $s,$  
\reject.\;

  \end{enumerate}
\nl \Accept.
\end{algorithm}

Note that \Cref{alg:diagonal_square_tester} could reject right after Step~\ref{step:alg-2-step-3}
if any diamond in $A$ has a black pixel on its fence. However, we skipped this operation to keep the algorithm concise, since it is not needed to provide the guarantees of \Cref{thm:connectedness_tester}.

The success probability of \Cref{alg:diagonal_square_tester} is analyzed in \Cref{lem:border-connectedness-success_probability}. The intuition behind the analysis is that every square $s$ that has a high local cost (i.e., is far from border-connectedness) has either many black pixels in the diamonds of $A$ or many black connected-components in the diamonds of $B$. This is because we can ``fix'' $s$ by making all black pixels in $A$ white and by connecting all black connected components of $B$ to the border at the cost of recoloring at most $m_i$ white pixels per connected component.

Recall from Definition~\ref{def:Grid_pixels_squares_of_different_levels} that a witness is a square of one of the levels that violates border-connectedness.

\begin{lemma}
\label{lem:border-connectedness-success_probability}
Fix level  $i\in\integerset{\log\frac{1}{\eps}}$.
Let $s\in \mathcal{S}_{i}$ be a witness that consists of pixels
$[k_{i}]^{2}+(u,v)$.
A border-connectedness subroutine called by
Algorithm~\ref{alg:connectedness_tester} rejects $s$ with
probability at least $\frac{\elc(s)\cdot\alpha}{2k_{i}}$, where $\alpha=1$ for \emph{Exhaustive-Square-Tester} and $\alpha=1-e^{-1}$ for \emph{Diagonal-Square-Tester}.
\end{lemma}
\begin{proof}
\emph{Exhaustive-Square-Tester} determines that $s$ is a witness with
probability $1\geq \frac{\elc(s)}{2k_{i}}\cdot 1$. 

Now we prove the statement
for \emph{Diagonal-Square-Tester}. Let $A$ and $B$ be defined as in Algorithm~\ref{alg:diagonal_square_tester} after Step~3.
Let $\nA$ be the number of black pixels in all the diamonds of the set $A$ and on the fences of these diamonds.
Consider the connected components of the image graph that are formed
by black pixels in all the diamonds in the set $B$ and in their fences. Let $\cal{K}$ be the set of all such connected components that contain no pixels on the border of $s$, and define $\nB=|\cal{K}|.$ 

Next, we prove that 
\begin{align}\label{eq:diamond-connect-cost}
\nB m_{i}+\nA \geq \lc(s) \geq \elc(s).
\end{align}
The second inequality in \eqref{eq:diamond-connect-cost} holds by \Cref{def:elc}. To prove the first, we will show how to obtain a square in $\C'$ by modifying at most $\nB m_i+\nA$ pixels in $s$.

We claim that we can connect all $\nB$ connected components to each other and to the border of the square by modifying at most $\nB m_{i}$ pixels. To see this, notice that any two black pixels in the same diamond 
can be connected to each other by changing less than $m_i$ pixels (by taking any Manhattan-distance shortest path between the two pixels and making it all black). To prove the claim, we can connect the $\nB$ connected components to the border of the square in the order their diamonds were added to $B$ by the algorithm. The initial diamonds placed in $B$ have at least one pixel on the border of the square, so their connected components can be directly connected to the border (using at most $m_i$ pixels per connected component). Now assume that we already connected to the border all components that have pixels in the diamonds added to $B$ so far. When the algorithm moves some diamond $d_2$ from $A$ to $B$, it is done because there is already a diamond $d_1$ in $B$ such that $d_1$ and $d_2$ have a black pixel $\beta$ in the common portion of their fences.  Then $\beta$ must be already connected to the border. If there are any connected components in $d_2$ that are not connected to the border yet, we can fix that by connecting them to $\beta$ (using at most $m_i$ pixels per connected component). We proceed like this until all $\nB$ connected components of the diamonds that were added to $B$ by the algorithm are connected to the border of the square. At this point, we have changed at most $\nB m_i$ pixels.

Finally, we color white all the $\nA$ black pixels in the diamonds of the set $A$ and on their fences. 
Altogether, after these at most $\nB m_i+\nA$ modifications to $s$, we obtain a square in
$\C'$. Thus, \eqref{eq:diamond-connect-cost} holds.

Next, we use \eqref{eq:diamond-connect-cost} to analyze the success probability of \Cref{alg:diagonal_square_tester}. Observe that for $x\in[0,1]$,
\begin{align}\label{eq1}
x\geq 1-e^{-x}\geq x(1-e^{-1})> x/2. 
\end{align}
Consider one iteration of Step~\ref{step:sample-before-bfs} of Algorithm~\ref{alg:diagonal_square_tester}. Recall that $p$ is the pixel selected in that step. Let $E_1$ be the event that $p$ is one of the $\nA$ black pixels in the diamonds of the set $A$.
Let $E_2$ be the event that $p$ is in
one of the 
connected components 
in $\cal{K}$
and that its connected component is completely discovered by the BFS in Step~\ref{step:sample-before-bfs}\ref{step:BFS}. Then the probability that one iteration  
of Step~\ref{step:sample-before-bfs} of 
Algorithm~\ref{alg:diagonal_square_tester} rejects square $s$ is equal to $\Pr[E_1]+\Pr[E_2]$. 
Since $p$ is chosen uniformly from a square of size $k_i\times k_i$, $$\Pr[E_1]= \frac \nA {k^2_i}.$$
Since $p$ and $x$ are chosen independently,
    $$\Pr[E_2]
    = \sum_{C\in \cal{K}} \Pr_p[p\in C]\cdot \Pr_x\big[x\geq |C|\big]
    = \sum_{C\in \cal{K}} \frac{|C|}{k^2_i}\cdot \frac 1 {|C|}
    =\sum_{C\in \cal{K}}\frac 1 {k^2_i}
    =\frac \nB {k^2_i}.$$

Then, by \eqref{eq:diamond-connect-cost} and \eqref{eq1} and since $\frac{elc(s)}{2k_i}\leq 1$ by \Cref{def:elc}, the probability that all $\frac1 2 k_im_i$ iterations  of Step~\ref{step:sample-before-bfs} of Algorithm~\ref{alg:diagonal_square_tester} complete 
without rejecting is

 $$\left(1-\frac{\nA+\nB}{k^2_i}\right)^{\frac1 2k_im_i}\leq e^{-\frac{\nA+\nB}{2k_i}\cdot m_i}\leq e^{-\frac{\nA+\nB m_i}{2k_i}}\leq e^{{-\frac{elc(s)}{2k_i}}}\leq 1-\frac{elc(s)}{2k_i}(1-e^{-1}).$$ 
Therefore, Algorithm~\ref{alg:diagonal_square_tester}  rejects $s$ with probability at least $$\frac{elc(s)}{2k_i}(1-e^{-1}),$$
completing the proof of \Cref{lem:border-connectedness-success_probability}
\end{proof}

\subsection{Proof of \Cref{thm:connectedness_tester}}
In this section, we use Lemmas~\ref{lem:sum_of_local_costs} and~\ref{lem:border-connectedness-success_probability} to complete the analysis of our connectedness tester and the proof of \Cref{thm:connectedness_tester}.

\Cref{alg:connectedness_tester} always accepts connected images,
since it sees no violation of $\C'$ in Step~\ref{step:square-test}\ref{step:test-for-border-connectedness}. Consider an image $M$ that is $\eps$-far from connectedness. Fix level $i\in\big[0..\log\frac 1 \eps\big)$ and one iteration of Step~\ref{step:square-test} of \Cref{alg:connectedness_tester}.
 Let $s$ be a square of level $i$ and $E_s$ be the event that,

 in this iteration, $s$ is selected in Step~\ref{step:square-test}\ref{step:sample-squares} and rejected by the subroutine in Step~\ref{step:square-test}\ref{step:test-for-border-connectedness}.
 Observe that for every $i\in\integerset{\log\frac{1}{\eps}}$, there are $\frac{n^{2}}{k_{i}^{2}}$ squares in $\mathcal{S}_{i}$. By \Cref{lem:border-connectedness-success_probability}, $$\Pr[E_s]\geq\frac{k_{i}^{2}}{n^{2}}\cdot\frac{\alpha \cdot elc(s)}{2k_i}=\frac{\alpha\cdot k_i\cdot elc(s)}{2n^2}.$$
 Let $E$ be the event that,
 in this iteration, the subroutine in Step~\ref{step:square-test}\ref{step:test-for-border-connectedness} rejects.
Since events $E_s$ are disjoint and $E=\bigcup_{s\in\mathcal{S}_i} E_s,$  
$$\Pr[E]=\sum_{s\in\mathcal{S}_i} \Pr[E_s]\geq\frac{\alpha\cdot k_i}{2n^2}\cdot \sum_{s\in\mathcal{S}_{i}} elc(s).$$ 
Let $\mathcal S$ be the set of squares of all levels $i\in[0..\log\frac 1 \eps)$. Recall that $k_i=\frac{4}{\eps}\cdot 2^{-i}-1$ and $\frac 1 {\eps}$ is a power of 2. Then, $\log \frac 1{\eps}$ is an integer and $i\leq \log \frac 1{\eps}-1$. Consequently, $k_i\geq \frac{3.5}{\eps}\cdot 2^{-i}$. By independence, the fact that $1-x\leq e^{-x}$ for all $x$, and \Cref{lem:sum_of_local_costs}, the
overall probability that the subroutine accepts in all iterations is at most

\begin{align*}
\prod_{i\in\integerset{\log\frac{1}{\eps}}} \left(1-\frac{\alpha\cdot k_i}{2n^2}\cdot \sum_{s\in\mathcal{S}_{i}} elc(s)\right)^{2^{i+1}}
&\leq \prod_{i\in\integerset{\log\frac{1}{\eps}}} \exp\left(-\frac {\alpha}2 \cdot \frac{3.5}{\eps n^2 2^i}\cdot \sum_{s\in\mathcal{S}_{i}} elc(s)\cdot 2^{i+1}\right)\\
=\exp\left(-\frac{3.5\alpha}{\eps n^2}\cdot \sum_{s\in\mathcal{S}} els(s)\right)
&\leq \exp\left(-\frac{3.5\alpha}{\eps n^2}\cdot \frac{\eps n^2}{2}\right)=e^{-1.75\alpha}.
\end{align*}

Therefore, the probability that \Cref{alg:connectedness_tester} detects at
least one witness 
is at least $1-e^{-1.75\alpha}$. Since at least an $\eps$
fraction of pixels in $M$ are black and every square of every level contains at most
$\frac{16}{\eps^{2}}\leq\frac{\eps n^{2}}{4}$ pixels (recall that $n\geq 8\eps^{-3/2}$ in the premise of \Cref{thm:connectedness_tester}), the probability that Algorithm~\ref{alg:connectedness_tester} detects a black pixel outside of that witness in Step 1 is at least
$1-(1-\frac{3\eps}{4})^{\frac{8}{\eps}}>1-e^{-6}$. Thus, for both values of $\alpha$, the probability
that Algorithm~\ref{alg:connectedness_tester} rejects $M$ is at least $$(1-e^{-6})(1-e^{-1.75\alpha})\geq 2/3,$$
completing the analysis of the success probability of \Cref{alg:connectedness_tester}.

\paragraph{Query Complexity.} 
Step~\ref{step:query-pixels} of Algorithm~\ref{alg:connectedness_tester} makes $O(\frac 1 \eps)$ queries.
We first prove that Algorithm~\ref{alg:connectedness_tester} has query
complexity $O(\frac{1}{\eps^{2}})$ if it uses \emph{Exhaustive-Square-Tester}
as a subroutine. In Steps~\ref{step:go-through-all-i}--\ref{step:square-test},
Algorithm~\ref{alg:connectedness_tester} samples $2^{i+1}$ squares of each level
$i\in[0..\log\frac{1}{\eps})$ and, for each sampled square, it calls
\emph{Exhaustive-Square-Tester} which makes
$(\frac{4}{\eps}\cdot2^{-i}-1)(\frac{4}{\eps}\cdot2^{-i}-1)<\frac{16}{\eps^{2}2^{2i}}$ queries in each sampled square of level $i$.
Thus, the query complexity of Steps~\ref{step:go-through-all-i}--\ref{step:square-test} is $$\sum\nolimits_{i=0}^{\log\frac{1}{\eps}} 2^{i+1}\cdot \frac{16}{\eps^{2}2^{2i}}<\sum\nolimits_{i=0}^{\log\frac{1}{\eps}}\frac{32}{\eps^{2}2^{i}}=O\left(\frac{1}{\eps^{2}}\right),$$
and the total query complexity of \Cref{alg:connectedness_tester} is also $O(\frac 1{\eps^2}).$

When Algorithm~\ref{alg:connectedness_tester} uses \emph{Diagonal-Square-Tester}, it queries at most
$\frac{2k_{i}^2}{m_{i}}$ diagonal lattice pixels inside each square of
level $i$ (in Step \ref{step:diagonal lattice pixels} of the subroutine). After that, in Step~\ref{step:sample-before-bfs} of the subroutine, it selects $\frac 1 2{k_{i}m_{i}}$ pixels and a number
$x\in[k_{i}^{2}]$ from the specified distribution and then makes at
most $4x$ queries for each selected pixel. Observe that $\E[x]=O(\log k_{i})$. Thus, the expected number of queries inside a square of level $i$ is at most
$\frac{2k_{i}^2}{m_{i}}+\frac 1 2k_{i}m_{i}\cdot 4\cdot O(\log
k_{i})=O(k_{i}^{3/2}\sqrt{\log k_{i}})$. The expected total number of queries is
$\sum\nolimits_{i=0}^{\log(1/\eps)}O(k_{i}^{3/2}\sqrt{\log
k_{i}})\cdot 2^{i+1}=O(\eps^{-3/2}\sqrt{\log \frac{1}{\eps}})$.

By standard arguments, the adaptive version of Algorithm~\ref{alg:connectedness_tester} can be converted to an algorithm that makes asymptotically the same number of queries in the worst case,  and has the same accuracy guarantee and running time.

\paragraph{Running Time.} 
To analyze the running time, we assume that individual queries and basic arithmetic operations can be performed in constant time.
The time 
complexity of Step 1 of Algorithm~\ref{alg:connectedness_tester} is $O(\frac{1}{\eps})$.
Therefore, the total time complexity of Algorithm~\ref{alg:connectedness_tester} is
$O(\frac{1}{\eps})$ plus the time complexity of Steps~\ref{step:go-through-all-i}--\ref{step:square-test}. In Step~\ref{step:square-test}\ref{step:test-for-border-connectedness}, Algorithm~\ref{alg:connectedness_tester}
uses either \emph{Exhaustive-Square-Tester} or \emph{Diagonal-Square-Tester}.

\emph{Exhaustive-Square-Tester} can be implemented by running a BFS on the image graph for the sampled square to find a proof of disconnectedness.
BFS takes time that is linear in the sum of the number of edges and the number of nodes of the graph.
Every pixel of a sampled square has at most $4$ neighboring pixels.
Thus, the number of edges in the image graph of every sampled square is linear in the number of pixels inside it and the time complexity of
Steps~\ref{step:go-through-all-i}--\ref{step:square-test} is linear in the number of all queried pixels, i.e.,
$O(\frac{1}{\eps^{2}})$ when \emph{Exhaustive-Square-Tester} is used.

Next, we argue that \emph{Diagonal-Square-Tester} takes time linear in the number of pixels it queries. The number of diamonds it considers is at most linear in the number of queried pixels. Step~\ref{step:alg-2-step-3}, where \emph{Diagonal-Square-Tester} partitions the diamonds into sets $A$ and $B$, can be implemented by running a BFS 
on the graph where the nodes correspond to the diamonds, and two diamonds are adjacent if there is a black pixel on their shared border. This BFS takes linear time because the graph has constant degree. In Step~\ref{step:sample-before-bfs}, we check for every queried black pixel what type of diamond it is in. This check can be performed in $O(1)$ time. If the pixel is in a diamond from set $B$, this pixel is used as part of a BFS in Step~4\ref{step:BFS} on the image graph. As we argued before, a BFS on the image graph takes time linear in the number of pixels queried by the BFS.
Thus, the time complexity of
Steps~\ref{step:go-through-all-i}--\ref{step:square-test} of Algorithm~\ref{alg:connectedness_tester} is linear in the number of all queried pixels, i.e.,
$O(\eps^{-3/2}\sqrt{\log (1/\eps)})$ when \emph{Diagonal-Square-Tester} is used. 
This completes the proof of \Cref{thm:connectedness_tester}.

\section{Lower Bound for Testing Connectedness}

In this section, we give a lower bound on the query complexity of testing connectedness, proving
\Cref{thm:nonadap-bound-connectedness}. We use the standard setup of constructing a distribution $\dN$ on $\eps$-far inputs such that every deterministic nonadaptive algorithm that makes $q\leq\frac c \eps \log \frac 1 \eps$ queries (for some constant $c$) has probability of error greater than 1/3. Then \Cref{thm:nonadap-bound-connectedness} follows by Yao's Principle~\cite{Yao}.

\subsection{Construction of the Hard Distribution}\label{sec:construction}

\begin{figure}[ht]
\centering
\includegraphics[width=0.35\linewidth]{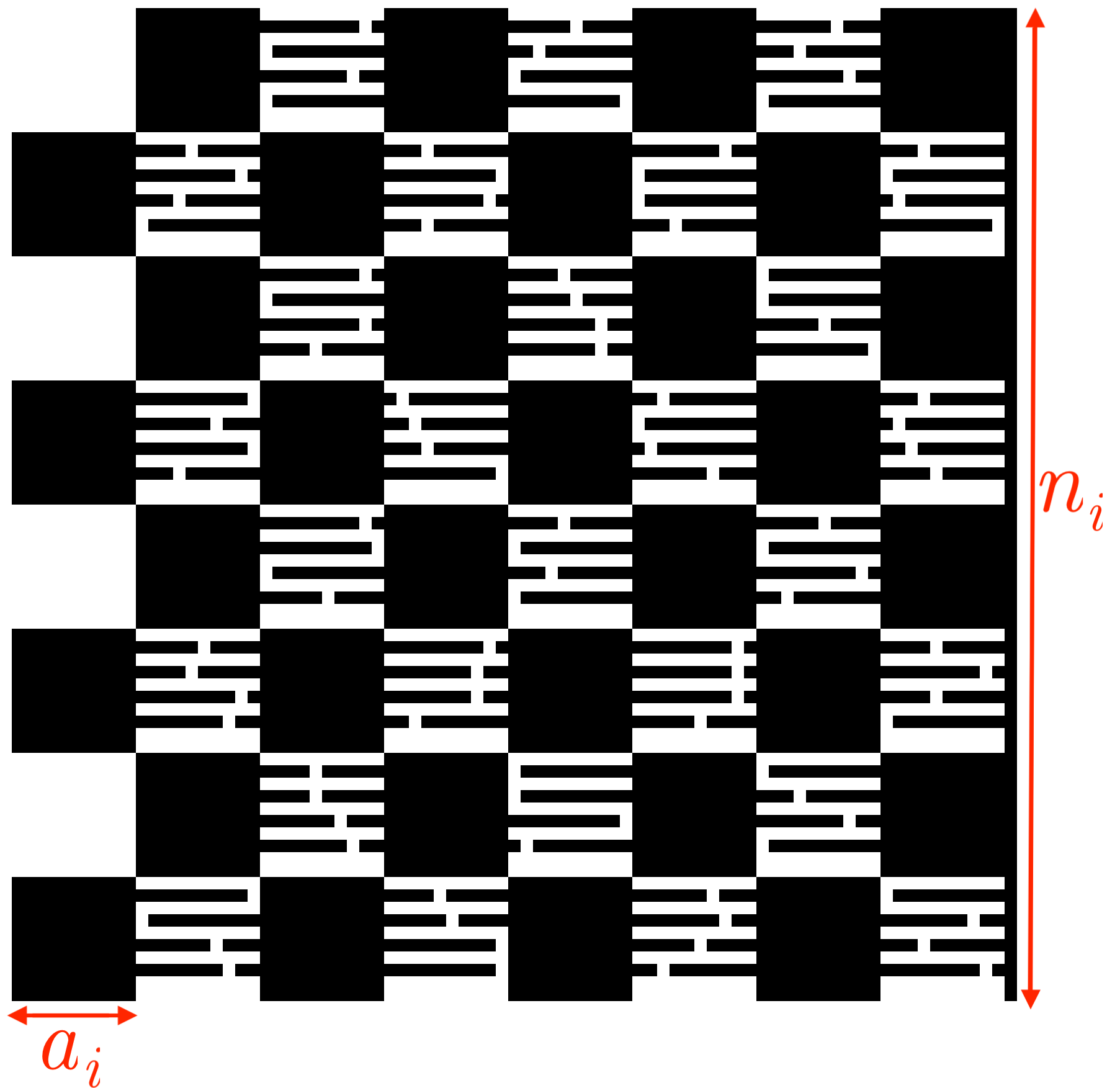}
\caption{Our construction of distribution $\dN$: an interesting window together with a column of black pixels immediately to the right of it. All other pixels in the constructed image are white.}
\label{fig:interesting-square}
\end{figure}

 The construction is parameterized by a positive integer $n$ and the proximity parameter $\eps\in(0,1)$. It gives a distribution supported on $(n+1)\times(n+1)$ images that are $\eps$-far from connectedness. We assume that 
 $n>\frac 1{16}(\frac 1\eps)^{5/8}.$ 
 
 We also assume that $n$ is a power of $2$ and $1/\eps$ is an even power of $2$ and that both of them are sufficiently large, so that all indices in our construction are integers. (See also Footnote~\ref{footnote:integrality-issues} for the discussion of integrality issues.) 
Our construction starts by selecting a {\em level}. The indices of the levels range from the low index $\ell=\frac 1 8\log \frac 1 \eps$ to the high index $h=\frac 1 4\log \frac 1 \eps$. For all indices $i\in \{\ell,\ell+1,\dots, h\}$, define $a_{i}=2^i$ and $n_{i} = 16 \cdot \sqrt{\epsilon} \cdot n \cdot a_{i} = \sqrt{\epsilon} \cdot n \cdot 2^{i+4}$.
First, we pick a uniformly random integer 
$i\in \{\ell,\ell+1,\dots, h\}$.
Consider the $n \times n$ image resulting from removing the last row and the last column of the $(n+1)\times(n+1)$ image. We partition this $n\times n$ image into $(\frac{n}{n_i})^2$ squares with side length $n_i$ called {\em windows of level $i$}; that is, each window of level $i$ is an $n_i\times n_i$ subimage. We pick one of the windows of level $i$ uniformly at random and call it an {\em interesting window}. We make the $n_i$ pixels immediately to the right of the interesting window black,  representing a vertical black line segment, and all other pixels  outside of the interesting window white.
 
 Now we describe how to color the interesting window. 
 
 See Figure~\ref{fig:interesting-square} for an illustration.
 
We  fill the interesting window with a checkerboard pattern 
consisting of 
$a_{i} \times a_{i}$ squares.
Inside each white checkerboard square that is not in the first column, number the rows starting from 0 and make every odd row, 
excluding the last, fully black, except for one randomly selected pixel for each row. The resulting $\frac{a_i}{2}-1$ black lines, each containing one white pixel, are called {\em bridges}. The white pixels on the bridges are called {\em disconnecting pixels.} We refer to each checkerboard square with the bridges as a {\em bridge square}. 

The intuition behind the construction is the following. Each black checkerboard square is in its own connected component. However, to ``catch'' this connected component as a witness of disconnectedness, a tester would have to query all the disconnecting pixels in the bridge squares to the left and/or to the right of the black square. Since the positions of the disconnecting pixels are random, it would have to query $\Omega(a_i^2)$ pixels in at least one relevant bridge square. To catch a witness of disconnectedness for a specific level, the algorithm would have to do it for many windows of that level because the interesting window is positioned randomly.
The key feature of our construction is that interesting windows of different levels are either disjoint or contained in one another, so potential witnesses for one window cannot significantly help with another. More precisely, we prove that a tester must query $\Omega(\frac{1}{\eps})$ pixels per level. Since there are $\Theta(\log \frac{1}{\eps})$ levels, this leads to the $\Omega(\frac{1}{\eps} \log{\frac{1}{\eps}})$ final query complexity. 

Note that it is essential that the lower bound is for nonadaptive testers, since an adaptive tester could use $O(\frac{1}{\eps})$ queries to identify the level and the position of the interesting window, followed by $O(\frac{1}{\eps})$ queries to check one of the bridges inside the interesting window.

\subsection{Analysis of the Construction}

We start by showing that all images in the support of $\dN$ are far from connected.

\begin{lemma} Every image in the support of $\dN$ is $\eps$-far from connected.
\end{lemma}
\begin{proof}
For every $i \in\{\ell, \ldots, h\}$, there are $\frac{(n_i / a_i)^2}{2} = 128 \epsilon n^2$ black regions. Each of them is in a separate connected component, except for the last $\frac{n_i / a_i}{2}=8\sqrt{\epsilon} \cdot n$, which are connected by the vertical line. Thus, the image graph has at least $
128 \cdot \epsilon n^2 - 8\sqrt{\epsilon} \cdot n$ connected components.
 Changing one pixel from white to black corresponds to adding a node of degree at most 4 to the image graph. This decreases the number of connected components by at most 3. Removing a pixel decreases the number of connected components by at most 1. 
Consequently, overall, we need to change at least $\frac 13{(128\epsilon n^2 - 8\sqrt{\epsilon} \cdot n)}$ pixels.
This is at least $\epsilon n^2$ for sufficiently large $n$; in particular, it holds for $n>\frac 1{16}(\frac 1\eps)^{5/8}$, which was our assumption in the beginning of \Cref{sec:construction}.   
\end{proof}
Next we show that if the number of queried pixels is small, then every 1-sided error deterministic algorithm detects a violation of connectedness in an image distributed according to $\dN$ with insufficiently small probability.

\begin{lemma}\label{lm:hard-to-find-witness}
Let $M$ be an image distributed according to $\dN$. Fix a deterministic nonadaptive 1-sided error algorithm $\mathcal A$ for testing connectedness of images. Let $\bf Q$ be the set of pixels queried by $\mathcal A$ and let $q=|\bf Q|$. For sufficiently small constant $c>0$, if $q\leq \frac c{\eps}\log \frac1 {\eps}$, then $\mathcal A$ detects a violation of connectedness in $M$ with probability less than $1/3$.

\end{lemma}
\begin{proof}
For each level $i\in[\ell,h]$, consider the partition of the image into squares with side length $a_i$ and call each such square a {\em cell} of level $i$. Note that when image $M$ is drawn from $\dN$, the bridge squares of level $i$ are selected from cells of level $i$.
An important feature of our construction is that the largest cell is smaller than the smallest window. Indeed, the side length of the largest cell is $a_h=(\frac 1\eps)^{1/4}$, whereas the side length of the smallest window is $n_\ell= 16\sqrt{\eps}\cdot (\frac 1 \eps)^{1/8}n=16 \eps^{3/8}n$. Thus, $a_h<n_\ell$, as long as  $n>\frac 1{16}(\frac 1\eps)^{5/8}$, as we assumed in the beginning of \Cref{sec:construction}. The consequence of this feature and the fact that both $a_i$'s and $n_i$'s are powers of 2 is that each cell of level $i$ is contained in one window of level $j$ for all levels $i,j\in \{\ell,\ell+1,\dots,h\}.$

Recall that set $\bf Q$ of queries is fixed, and that bridge squares are chosen randomly when the image $M$ is drawn from $\dN$. Intuitively, in order to certify that $M$ is disconnected, the tester has to query many points in some bridge square. To accomplish this, $\bf Q$ might contain many queries from some cells. However, when $\bf Q$ contains a small number of queries (specifically, at most $\frac c{\eps}\log \frac1 {\eps}$ queries), only a small number of cells can contain many queries.
We classify windows as good or bad based on whether they contain a cell with sufficiently many queries in $\bf Q$. We give up on good windows in the sense that we analyze the failure probability of the algorithm only for the case when a bad window is chosen as an interesting window. We show that the probability of failure for that case is already too high.

\begin{definition}[Covered cells, good windows]
A cell of level $i\in \{\ell,\ell+1,\dots,h\}$ is {\em covered} if ${\bf Q}$ contains at least ${a_i^2}/8$ pixels from that cell. A window of level $i$ is {\em good} if it contains a covered cell of level at least $i$; otherwise, it is {\em bad}.\footnote{We defined a good window of level $i$ with respect to covered cells of level $i$ or above, as opposed to just level $i$. Observe that if a cell of level $i+1$ is covered, then at least one of the four cells of level $i$ contained in the considered cell is also covered. Consequently, if a window is good, it contains a covered cell of its own level. While our choice of ``level $i$ or above'' in the definition does not change its substance compared to just ``level $i$'', it highlights that larger cells can make a window good. This is important in the description and analysis of the association procedure in the next part of the proof.}
\end{definition}

The difficulty with counting queries needed for testing images from $\dN$ is that some queries can help with covering different cells. To resolve this issue, we define {\em maximal} cells and give a procedure that associates many good windows with maximal cells. We call a cell {\em maximal} if it is covered and not contained in another covered cell. Now we associate some of the good windows with unique maximal cells. To do this, {\em we go through levels from highest to lowest, and for each level $i\in[\ell,h]$, we consider good windows in an arbitrary order. If for a considered window $w$, all of the maximal cells of level at least $i$ inside $w$ are not associated with any previously considered window, we associate $w$ with an (arbitrary) one of those maximal cells.
This association allows us to avoid overcounting queries needed when considering different levels.} All windows of the same level are associated with different covered cells, because each cell is contained in exactly one window of a given level. 

\begin{claim}[Guarantees of the Association Procedure]\label{claim:properties-of-good-windows}
For every level $i\in \{\ell,\ell+1,\dots,h\},$ let $G_i$ be the set of good windows of level $i$ associated with 
maximal cells, let $T_i$ be the set of all good windows of level $i$, and define $g_i=|G_i|$, and $t_i=|T_i|.$
Then the following hold:
\begin{enumerate}
\item\label{item:lb-on-q} 
$\displaystyle q\geq \sum_{i=\ell}^h \frac{a^2_i g_i}8$, where $q=|{\bf Q}|.$

\item\label{item:recursion} $g_h=t_h$ and $g_i \geq t_i-t_{i+1}$ for all $i\in [\ell,h-1].$ 
\end{enumerate}
\end{claim}
\begin{proof} We prove the two items separately.
\begin{enumerate}
    \item For all distinct levels $i,j\in[\ell,h]$, the sets $G_i$ and $G_j$ are disjoint because they contain windows of different levels. All windows in $\bigcup_{i\in[\ell,h]}G_i$ are associated with different maximal cells by construction of associations. Moreover, for all levels $i\in[\ell,h]$, each window in $G_i$ is associated with a maximal cell of level at least $i$, that is, a covered cell of that level. By definition, such a cell contributes at least $a^2_i/8$ queries to $\bf Q.$

    \item To prove $g_h=t_h$, note that each maximal cell of level $h$ is contained in exactly one window of $G_h$. Thus, all windows in $T_h$ are associated with maximal cells, implying $G_h=T_h$, and consequently $g_h=t_h$. 

    For the next part of the proof, we use the following terminology. For every window $w$ of a level $i\in[\ell,h-1]$, the {\em parent} of $w$, denoted $p_w$, is the window  of level $i+1$ that contains $w$. Window $w$ is called a {\em child} of $p_w$. In our construction, each window of level  $i\in[\ell+1,h]$ has 4 children.

    Now we fix a level $i\in[\ell,h-1]$ and argue that $t_i\leq g_i+t_{i+1}$. The sets $G_i$ and $T_{i+1}$ are disjoint because they contain windows of different levels. Therefore, $|G_i\cup T_{i+1}|=g_i+t_{i+1}.$ To prove $t_i\leq g_i+t_{i+1}$, we establish an injection from $T_i$ to $G_i\cup T_{i+1}$. Each window of $T_i$ that is also in $G_i$ is mapped to itself, and each window $w$ in $T_i\setminus G_i$ is mapped to $p_w$, i.e., the parent of~$w$.
    
    Next, we assume that $w\in T_i\setminus G_i$ and argue that $p_w\in T_{i+1}$. Since $w$ is a good window of level $i$, it contains a covered cell of level at least $i$. Recall that the largest cell is smaller than the smallest window, and that a window cannot contain only a part of a cell (because the side lengths of windows and cells in our partitions of the image are powers of two). As a result, $w$ also contains a maximal cell. Since $w$ is in $T_i$, but not in $G_i$, one of the maximal cells from $w$ is associated with some window $w'$ that was considered before $w$ in the construction of associations. Call this maximal cell $c_w$. Since each cell belongs to exactly one window of  level $i$, window $w'$ must be of level at least $i+1.$ Therefore, $c_w$ is a covered cell of level at least $i+1$. Since cell $c_w$ belongs to window $w$, it also belongs to window $p_w$. Consequently, $p_w$ is good because it is a window of level $i+1$ and contains a covered cell of level at least $i+1$. So, $p_w\in T_{i+1}.$  

    Finally, we show that our mapping from $T_i$ to $G_i\cup T_{i+1}$ is injective. Consider windows $w_1,w_2\in T_i$. Suppose for the sake of contradiction that they are mapped to the same window $w$ in $G_i\cup T_{i+1}$. Since $w_1$ and $w_2$ are windows of the same level, and we map each of them either to itself or its parent, both $w_1$ and $w_2$ must be children of $w$. Since $w_1$ and $w_2$ were mapped to their parent, they are not associated with any maximal cells. It means that both of them contain maximal cells associated with windows of higher levels. Let $c_1$ and $c_2$ be such maximal cells in $w_1$ and in $w_2$, respectively. Suppose w.l.o.g.\ that $c_1$ was associated first. Let $w'$ be the window associated with $c_2$. The window $w'$ must be of level at least $i+1$, so it contains both $c_1$ and $c_2$. But the procedure that constructed associations could associate $w'$ with $c_2$ only if all maximal cells of the currently considered level inside $w'$ are not associated with any of the previously considered windows. Since $c_1$ was already associated, we get a contradiction.\qedhere
\end{enumerate}
\end{proof}

Next, we define an event $E$ that must happen in order for the algorithm $\A$ to succeed in finding a witness of disconnectedness in an image distributed according to $\dN$. Then we show that the probability of $E$ is too small. 

\begin{definition}[A revealing set and event $E$] 
Let $M$ be an image distributed according to $\dN$. 
The set of all disconnecting pixels from one bridge square of $M$ is called a {\em revealing set} for the window containing the bridge square. Let $E$ denote the event that ${\bf Q}$ contains a revealing set for the interesting window of $M$.

\end{definition}

Since the tester $\A$ has 1-sided error, it can reject only if it finds a violation of connectedness. In particular, 
to succeed on an image distributed according to $\dN$, algorithm
$\A$ must find a revealing set for the interesting window of the image.

\begin{claim}\label{cl:small-witnesses} 
If the number of queries  $q\leq\frac c \eps \log \frac 1 \eps$ for sufficiently small constant $c$, then 
$\Pr[E]<1/3.$
\end{claim}

\begin{proof}
Let $G$ be the event that the interesting window in $M$ is good. Then, by the law of total probability,
$$\Pr[E]
=\Pr[E\mid G]\cdot\Pr[G]+\Pr[E\mid \overline{G}]\cdot\Pr[\overline{G}]
\leq \Pr[G]+\Pr[E\mid \overline{G}].$$
Next, we analyze the event $G$. 

We begin by giving a lower bound on $q$, the number of queries in $\bf Q$ in terms of the number of good windows of different levels. We first apply Items~\ref{item:lb-on-q} and \ref{item:recursion} of \Cref{claim:properties-of-good-windows} and then rearrange the terms of the resulting summation:

\begin{align}\label{eq:q-bound}
q
&\geq \sum_{i=\ell}^h \frac{a^2_i g_i}8 
\geq \frac{1}{8}\left(\sum_{i=\ell}^{h-1} a^2_i (t_i-t_{i+1}) + a_h^2 t_h \right)\nonumber\\
&= \frac{1}{8}\left(\sum_{i=\ell}^{h-1} 4^i (t_i-t_{i+1}) + 4^h t_h \right)\nonumber
= \frac{1}{8}\left(\sum_{i=\ell}^{h} 4^i t_i-\sum_{i=\ell+1}^{h}4^{i-1}t_{i} \right)
\\
&\geq \frac{3}{32} \sum_{i=\ell}^{h} 4^{i}t_i.
\end{align}
Recall that $q\leq \frac c{\eps}\log \frac1 {\eps}$, that the total number of levels is $h-\ell+1=\Theta(\log\frac 1\eps)$, and that the number of all windows for each level $i$ is 
$(\frac n{n_i})^2=\Theta(\frac 1 {4^i\eps}).$ 
When $M$ is drawn from $\dN$, first a level $i$ is chosen uniformly at random, i.e., each level is selected with probability $\frac 1 {h-\ell+1}$, and then an interesting window is chosen uniformly at random among the $(\frac n{n_i})^2$ available windows of level $i$. Since $t_i$ windows of level $i$ are good, the probability that a good window is chosen as an interesting window is
\begin{align*}
 \Pr[G]
& =\frac 1 {h-\ell+1}\sum_{i=\ell}^{h}\frac{t_i}{(n/n_i)^2}
=\frac 1 {\Theta(\log(\frac {1}{\eps}))} \sum_{i=\ell}^{h}(t_i\cdot \Theta({4^i\eps}))
=\frac {1} {\Theta(\frac 1\eps \log(\frac {1}{\eps}))} \sum_{i=\ell}^{h}(4^it_i)\\
&< \frac {q} {\Theta(\frac 1\eps \log(\frac {1}{\eps}))}
<1/6,   
\end{align*}
where \eqref{eq:q-bound} was used to obtain the first inequality, and the last inequality holds for sufficiently small constant $c$. 

It remains to analyze $\Pr[E\mid \overline{G}],$ that is, the probability of $E$, given that the interesting window is bad. Consider a bad window of level $i\in \{\ell,\dots, h\}$. It has at most $q$ bridge squares that contain a queried pixel. Consider one of such bridge squares.
Recall that this bridge square has $a_i/2-1$ bridges. Number these bridges using integers $1,2, \dots, a_i/2-1$. Let $x_k$ denote the number of pixels of {\bf Q} on the bridge number $k\in [a_i/2-1]$ of the bridge square. Then the probability over the choice of disconnecting pixels that this bridge square has a revealing set for the window is 
\begin{align*}
\nonumber    \prod_{k=1}^{a_i/2-1} \frac{x_k}{a_i}
&\leq \left(  \frac 1{a_i/2-1}\cdot\sum_{k=1}^{a_i/2-1}\frac{x_k}{a_i}\right)^{a_i/2-1}
\leq \left(  \frac {a_i/8}{a_i/2-1}\right)^{a_i/2-1}
\leq \left(\frac{1}{2}\right)^{a_i/2-1}\\
&\leq 2(\sqrt{2})^{-\sqrt[8]{\frac 1{\eps}}},
\end{align*}
where the first inequality follows from the inequality between geometric and arithmetic means, the second inequality holds since $\sum_{k=1}^{a_i/2}x_k<a_i^2/8$, and the third inequality holds since $\eps$ is sufficiently small and, consequently, we can assume that the minimum value of $a_i$ is at least 4.
By a union bound over all bridge squares in this window that contain a query, 
\begin{align*}
\Pr[E\mid \overline{G}]\leq 2(\sqrt{2})^{-\sqrt[8]{\frac 1{\eps}}}\cdot q
\leq
 2(\sqrt{2})^{-\sqrt[8]{\frac 1{\eps}}}\cdot \frac c{\eps} \cdot{\log \frac1 {\eps}<\frac 1 6},    
\end{align*}
for sufficiently small  $\eps$ and constant $c$. 
Thus, $\Pr[E]\leq \Pr[G]+\Pr[E\mid \overline{G}]<\frac 16+ \frac 1 6=\frac 1 3,$ 
as claimed. This completes the proof of Claim~\ref{cl:small-witnesses}.
\end{proof}

This concludes the proof of Lemma~\ref{lm:hard-to-find-witness}.
\end{proof}
Theorem~\ref{thm:nonadap-bound-connectedness} follows from Lemma~\ref{lm:hard-to-find-witness}.

\bibliography{visual-properties}

\begin{thebibliography}{KKNBZ11}

\bibitem[ABF17]{AlonBF17}
Noga Alon, Omri Ben{-}Eliezer, and Eldar Fischer.
\newblock Testing hereditary properties of ordered graphs and matrices.
\newblock In Chris Umans, editor, {\em 58th {IEEE} Annual Symposium on
  Foundations of Computer Science, {FOCS} 2017, Berkeley, CA, USA, October
  15-17, 2017}, pages 848--858. {IEEE} Computer Society, 2017.

\bibitem[BF18]{BF18}
Omri Ben{-}Eliezer and Eldar Fischer.
\newblock Earthmover resilience and testing in ordered structures.
\newblock In Rocco~A. Servedio, editor, {\em 33rd Computational Complexity
  Conference, {CCC} 2018, June 22-24, 2018, San Diego, CA, {USA}}, volume 102
  of {\em LIPIcs}, pages 18:1--18:35. Schloss Dagstuhl - Leibniz-Zentrum
  f{\"{u}}r Informatik, 2018.

\bibitem[BKM14]{BerenbrinkKM14}
Petra Berenbrink, Bruce Krayenhoff, and Frederik Mallmann{-}Trenn.
\newblock Estimating the number of connected components in sublinear time.
\newblock {\em Inf. Process. Lett.}, 114(11):639--642, 2014.

\bibitem[BKR17]{BKR17}
Omri Ben{-}Eliezer, Simon Korman, and Daniel Reichman.
\newblock Deleting and testing forbidden patterns in multi-dimensional arrays.
\newblock In Ioannis Chatzigiannakis, Piotr Indyk, Fabian Kuhn, and Anca
  Muscholl, editors, {\em 44th International Colloquium on Automata, Languages,
  and Programming, {ICALP} 2017, July 10-14, 2017, Warsaw, Poland}, volume~80
  of {\em LIPIcs}, pages 9:1--9:14. Schloss Dagstuhl - Leibniz-Zentrum
  f{\"{u}}r Informatik, 2017.

\bibitem[BMR19a]{BMR19-algorithmica}
Piotr Berman, Meiram Murzabulatov, and Sofya Raskhodnikova.
\newblock The power and limitations of uniform samples in testing properties of
  figures.
\newblock {\em Algorithmica}, 81(3):1247--1266, 2019.

\bibitem[BMR19b]{BMR19uniform}
Piotr Berman, Meiram Murzabulatov, and Sofya Raskhodnikova.
\newblock Testing figures under the uniform distribution.
\newblock {\em Random Struct. Algorithms}, 54(3):413--443, 2019.

\bibitem[BMR22]{BermanMR22}
Piotr Berman, Meiram Murzabulatov, and Sofya Raskhodnikova.
\newblock Tolerant testers of image properties.
\newblock {\em ACM Transactions on Algorithms (TALG)}, 18(4):1--39, 2022.

\bibitem[BMRR23]{BermanMRR23}
Piotr Berman, Meiram Murzabulatov, Sofya Raskhodnikova, and Dragos Ristache.
\newblock Testing connectedness of images.
\newblock In Nicole Megow and Adam~D. Smith, editors, {\em Approximation,
  Randomization, and Combinatorial Optimization. Algorithms and Techniques,
  {APPROX/RANDOM} 2023, September 11-13, 2023, Atlanta, Georgia, {USA}}, volume
  275 of {\em LIPIcs}, pages 66:1--66:15. Schloss Dagstuhl - Leibniz-Zentrum
  f{\"{u}}r Informatik, 2023.

\bibitem[BRY14]{BermanRY14}
Piotr Berman, Sofya Raskhodnikova, and Grigory Yaroslavtsev.
\newblock $l_p$-testing.
\newblock In David~B. Shmoys, editor, {\em Symposium on Theory of Computing,
  {STOC} 2014, New York, NY, USA, May 31 - June 03, 2014}, pages 164--173.
  {ACM}, 2014.

\bibitem[CRT05]{CRT05}
Bernard Chazelle, Ronitt Rubinfeld, and Luca Trevisan.
\newblock Approximating the minimum spanning tree weight in sublinear time.
\newblock {\em SIAM J. Comput.}, pages 1370--1379, 2005.

\bibitem[FN07]{FischerN07}
Eldar Fischer and Ilan Newman.
\newblock Testing of matrix-poset properties.
\newblock {\em Comb.}, 27(3):293--327, 2007.

\bibitem[GGR98]{GGR98}
Oded Goldreich, Shafi Goldwasser, and Dana Ron.
\newblock Property testing and its connection to learning and approximation.
\newblock {\em J. ACM}, 45(4):653--750, 1998.

\bibitem[GR02]{GR02}
Oded Goldreich and Dana Ron.
\newblock Property testing in bounded degree graphs.
\newblock {\em Algorithmica}, 32(2):302--343, 2002.

\bibitem[KKNBZ11]{KleinerKNB11}
Igor Kleiner, Daniel Keren, Ilan Newman, and Oren Ben-Zwi.
\newblock Applying property testing to an image partitioning problem.
\newblock {\em IEEE Trans. Pattern Anal. Mach. Intell.}, 33(2):256--265, 2011.

\bibitem[KRT11]{KormanRT}
Simon Korman, Daniel Reichman, and Gilad Tsur.
\newblock Tight approximation of image matching.
\newblock {\em CoRR}, abs/1111.1713, 2011.

\bibitem[KRTA13]{KormanRTA13}
Simon Korman, Daniel Reichman, Gilad Tsur, and Shai Avidan.
\newblock Fast-match: Fast affine template matching.
\newblock In {\em CVPR}, pages 2331--2338. IEEE, 2013.

\bibitem[MP17]{MinskyP70}
Marvin Minsky and Seymour~A. Papert.
\newblock {\em {Perceptrons: An Introduction to Computational Geometry}}.
\newblock The MIT Press, 09 2017.

\bibitem[Mur17]{Mur-thesis17}
Meiram Murzabulatov.
\newblock {\em Testing Geometric Properties of Two-Dimensional Figures and
  Images}.
\newblock PhD thesis, The Pennsylvania State University, 2017.

\bibitem[PCR19]{ParalCR19}
Pritam Paral, Amitava Chatterjee, and Anjan Rakshit.
\newblock Vision sensor-based shoe detection for human tracking in a
  human--robot coexisting environment: A photometric invariant approach using
  dbscan algorithm.
\newblock {\em IEEE Sensors Journal}, 19(12):4549--4559, 2019.

\bibitem[Ras03a]{Ras03}
Sofya Raskhodnikova.
\newblock Approximate testing of visual properties.
\newblock In Sanjeev Arora, Klaus Jansen, Jos{\'e} D.~P. Rolim, and Amit Sahai,
  editors, {\em RANDOM-APPROX}, volume 2764 of {\em Lecture Notes in Computer
  Science}, pages 370--381. Springer, 2003.

\bibitem[Ras03b]{Ras-thesis03}
Sofya Raskhodnikova.
\newblock {\em Property testing : theory and applications}.
\newblock PhD thesis, Massachusetts Institute of Technology, 2003.

\bibitem[RS96]{RS96}
Ronitt Rubinfeld and Madhu Sudan.
\newblock Robust characterizations of polynomials with applications to program
  testing.
\newblock {\em SIAM J. Comput.}, 25(2):252--271, 1996.

\bibitem[RT14]{TsurR10}
Dana Ron and Gilad Tsur.
\newblock Testing properties of sparse images.
\newblock {\em {ACM} Trans. Algorithms}, 10(4):17:1--17:52, 2014.

\bibitem[Yao77]{Yao}
Andrew~C. Yao.
\newblock Probabilistic computation, towards a unified measure of complexity.
\newblock In {\em Proceedings of the Eighteenth Annual Symposium on Foundations
  of Computer Science}, pages 222--227, 1977.

\end{thebibliography}
\bibliographystyle{alpha}
\end{document}